\documentclass[12pt]{article}
\usepackage{latexsym,amsmath}
\usepackage{graphicx}
\usepackage{caption}
\usepackage{float}
\usepackage{amsfonts}
\usepackage{isomath}
\usepackage{amssymb,array}
\usepackage{epsfig}
\usepackage{color}
\usepackage{setspace}
\usepackage{amsthm}
\usepackage{mathrsfs}
\usepackage[margin=1in]{geometry}
\usepackage{enumerate}
\usepackage{multicol}
\usepackage{multirow}
\usepackage{graphicx}
\usepackage{diagbox}
\usepackage{breqn}

\newtheorem{thm}{Theorem}[section]
\newtheorem{lem}{Lemma}[section]
\newtheorem{re}{Remark}[section]
\newtheorem{de}{Definition}[section]

\newtheorem{ex}{Example}[section]

\begin{document}

\title{On Weighted Mathai-Haubold Entropy Measures}\label{S1}

\author{$\text{Oindrali Das}^{\text{a}}, \text{Siddhartha Chakraborty}^{\text{b}}$\\
$^{\text{a}}$ Department of Statistics, Ranaghat College, Nadia, India \\
$^{\text{b}}$ Department of Statistics, University of Kalyani, Nadia, India}
\maketitle

\subsection*{Abstract}
In this paper, we propose weighted Mathai-Haubold entropy along with their residual and past versions and study their properties. We develop aging classes based on the weighted Mathai-Haubold residual and past entropy measures. Also, some inequalities related to the three proposed measures are discussed. Non-parametric estimators of the proposed measures are introduced and their properties are investigated. Performance of this estimator is evaluated by means of bias and mean squared error using Monte-Carlo simulations.\\

\noindent\textbf{Keywords :} Mathai-Haubold entropy, Aging classes, Inequalities, Kernel density estimator, Monte-Carlo simulation\\

\noindent \textbf{Mathematics Subject Classification:} 94A17, 62N05, 26D10

\section{Introduction}
Shannon \cite{1} introduced the concept of entropy to measure uncertainty. For a non-negative absolutely continuous random variable (rv) $X$ with probability density function (pdf) $g(x)$ and cumulative distribution function (cdf) $G(x)$ is defined as
\begin{eqnarray}\label{eq1}
H(X)=-\int_{0}^{\infty}g(x)\log{g(x)}dx.
\end{eqnarray}
Note that, ``$\log$'' is the natural logarithm and $0\log 0=0$ by convention. This entropy provides information on the expected uncertainty contained in $g(x)$ about the predictability of an outcome of $X$. Since then several generalizations of Shannon's entropy has been studied and are available in the existing literature, see, for example, R\'enyi \cite{2}, Tsallis \cite{3}, Rao et al. \cite{4} and the references therein.

Remaining lifetime of a rv is an important metric in reliability analysis. When an item has survived up to a time $t$, then Shannon entropy, defined in equation (\ref{eq1}), fails to provide information for the remaining lifetime of that item. In this respect, by applying Shannon information measure to the residual lifetime, Ebrahimi \cite{5} defined an alternative entropy for the residual lifetime $X_t=(X-t|X>t)$ of the rv $X$. This measure is known as residual entropy and is given by
\begin{eqnarray}\label{eq2}
H(g;t)=-\int_{t}^{\infty}\frac{g(x)}{\bar{G}(t)}\log{\frac{g(x)}{\bar{G}(t)}}dx,
\end{eqnarray}
where $\bar{G}(t)=P[X>t]=1-G(t)$ is the survival function of the rv $X$. It is to be noted that, when $t=0$, $H(g;t)$ reduces to $H(X)$. 

 The past lifetime of a rv $X$ is defines as ${_t}X=(t-X|X<t),\;t>0$, provided that the item has already failed at time $t$. Analogous to residual entropy, Di Crescenzo and Longobardi \cite{6} proposed past entropy measure as
\begin{eqnarray}\label{eq3}
\bar{H}(g;t)=-\int_{0}^{t}\frac{g(x)}{G(t)}\log{\frac{g(x)}{G(t)}}dx.
\end{eqnarray}
As $t\to\infty$, $\bar{H}(g;t)$ reduces to $H(X)$.

Recently, Mathai and Haubold \cite{7} introduced a new generalization of Shannon entropy measure associated with a continuous rv $X$, known as Mathai-Haubold entropy (MHE). The MHE of order $\alpha$ is defined as
\begin{eqnarray}\label{eq4}
M_{\alpha}(X)=\frac{1}{\alpha-1}\left[\int_{0}^{\infty}g^{2-\alpha}(x)dx-1\right],\;\alpha\neq1,0<\alpha<2.
\end{eqnarray}
As $\alpha\to1$, $M_{\alpha}(X)$ reduces to $H(X)$. The parameter $\alpha$ is called the gereralizing parameter which makes MHE more flexible than Shannon entropy.

Dar and Al-Zahrahi \cite{8} proposed Mathai-Haubold residual entropy (MHRE) by taking the current time of the unit into consideration. The MHRE is given by
\begin{eqnarray}\label{eq5}
M^{\alpha}(g;t)=\frac{1}{\alpha-1}\left[\int_{t}^{\infty}\frac{g^{2-\alpha}(x)dx}{\bar{G}^{2-\alpha}(t)}-1\right],\;\alpha\neq1,0<\alpha<2.
\end{eqnarray}


 Das et al. \cite{9} introduced the Mathai-Haubold past entropy (MHPE) measure which is defined in terms of the past lifetime of the rv $X$ and is given by
\begin{eqnarray}\label{eq6}
\bar{M}^{\alpha}(g;t)=\frac{1}{\alpha-1}\left[\int_{0}^{t}\frac{g^{2-\alpha}(x)dx}{G^{2-\alpha}(t)}-1\right],\;\alpha\neq1,0<\alpha<2.
\end{eqnarray}

These entropy only consider the probabilistic information of the underlying rv but do not consider the realization of the rv. Belis and Guiasu \cite{10} introduced the concept of weighted entropy by assigning non-negative weights to each event based on their utility. For the continuous case, weighted entropy is defined as 
\begin{equation}\label{e7}
H^w(X)=-\int_{0}^{+\infty}xg(x)\log g(x)dx,
\end{equation}
where the factor $x$ is the linear weight function that gives more importance to the larger values of the rv $X$.

In this work, we introduce weighted Mathai-Haubold entropy and its residual and past versions. We study some properties of the proposed measures. We develop aging classes based on weighted Mathai-Haubold residual and past entropy and study characterization result for Weibull distribution. Kernel based non-parametric estimators of these measures are discussed for independent and dependent data. The rest of the paper is organized as follows:

We introduced weighted Mathai-Haubold entropy (WMHE) and study its properties in Section \ref{S2}. Weighted Mathai-Haubold residual and past entropies are studied in Sections \ref{S3} and \ref{S4}, respectively. Some inequalities are discussed in Section \ref{S5}. Non-parametric estimators are proposed and their properties are investigated in Section \ref{S6}. Performance of the proposed estimators are evaluated by simulation for independent and $\rho$- mixing dependent data in Section \ref{S7}. Finally, some concluding remarks are made in Section \ref{S8}.

\section{Weighted Mathai-Haubold Entropy}\label{S2}
In this section, we propose weighted Mathai-Haubold entropy (WMHE) for a non-negative absolutely continuous rv $X$. 
\begin{de}
For a non-negative continuous rv $X$ with pdf $g(x)$, the WMHE is defined as
\begin{eqnarray}\label{eq7}
M^w_{\alpha}(X)=\frac{1}{\alpha-1}\left[\int_{0}^{\infty}x g^{2-\alpha}(x)dx-1\right], \;\alpha\neq1,0<\alpha<2.
\end{eqnarray}
\end{de}

WMHE for some popular distributions are provided in Table \ref{t1}.

\begin{table}[H]
\centering
\caption{WMHE for some distributions.}\label{t1}
\begin{tabular}{|c|c|c|}
\hline
\textbf{Distributions} &\textbf{$g(x)$} & \textbf{ $M^w_{\alpha}(X)$}\\
\hline
Uniform & $\frac{1}{b-a}; a<x<b$ & $\frac{1}{\alpha-1}\left[\frac{a+b}{2}(b-a)^{\alpha-1}-1\right]$\\
Exponential & $\lambda e^{-\lambda x}; x>0, \lambda >0 $ & $\frac{1}{\alpha-1}\left[\frac{\lambda^{-\alpha}}{(2-\alpha)^2}-1\right]$\\
Power & $\theta x^{\theta-1}; 0<x<1, \theta>0$ & $\frac{1}{\alpha-1}\left[\frac{\theta^{2-\alpha}}{(\theta-1)(2-\alpha)+2}-1\right]$ \\
\hline
\end{tabular}
\end{table}
The following example shows the effectiveness of WMHE measure.
\begin{ex}
If $X_1\sim U(a,b)$ and $X_2\sim U(a+h,b+h)$, then $M_{\alpha}(X_1)=M_{\alpha}(X_2)$ but $M^w_{\alpha}(X_1)\neq M^w_{\alpha}(X_2)$
\end{ex}
\begin{proof}
When $X_1\sim U(a,b)$,
\begin{eqnarray*}
M_{\alpha}(X_1)&=&\frac{1}{\alpha-1}\left[(b-a)^{\alpha-1}-1\right]\\
M^w_{\alpha}(X_1)&=&\frac{1}{\alpha-1}\left[\frac{a+b}{2}(b-a)^{\alpha-1}-1\right].
\end{eqnarray*}
Again when $X_2\sim U(a+h,b+h)$,
\begin{eqnarray*}
M_{\alpha}(X_2)&=&\frac{1}{\alpha-1}\left[(b-a)^{\alpha-1}-1\right]\\
M^w_{\alpha}(X_2)&=&\frac{1}{\alpha-1}\left[\frac{a+b+2h}{2}(b-a)^{\alpha-1}-1\right].
\end{eqnarray*}
\end{proof}

In the following we study the effect of linear transformation on WMHE measure.

\begin{lem}
Consider a linear transformation $Y=aX+b,\;a>0,b\geq 0$. Then 
\begin{eqnarray*}
(\alpha-1)M^w_{\alpha}(Y)+1=a^{\alpha}\int_{-\frac{b}{a}}^{\infty}x g_X^{2-\alpha}(x)dx+ a^{\alpha-1}b \int_{-\frac{b}{a}}^{\infty} g_X^{2-\alpha}(x)dx, \;\;\alpha\neq1,0<\alpha<2.
\end{eqnarray*}
\end{lem}
\begin{proof}
We want to express $M^w_{\alpha}(Y)$ in terms of $M^w_{\alpha}(X)$. Since $Y=aX+b$, the transformed density becomes 
\begin{eqnarray*}
g_Y^{2-\alpha}(y)=a^{\alpha-2} g_X^{2-\alpha}\left(\frac{y-b}{a}\right).
\end{eqnarray*}
Integrating and multiplying $y$ on both sides and substituting $x=\frac{y-b}{a}$ we get,
\begin{eqnarray*}
\int_0^{\infty} yg_Y^{2-\alpha}(y)dy &=& a^{\alpha-1}\int_{-\frac{b}{a}}^{\infty} (ax+b) g_X^{2-\alpha}(x) dx \\
                            &=& a^{\alpha-1} \left[ a\int_{-\frac{b}{a}}^{\infty}x g_X^{2-\alpha}(x)dx+ b \int_{-\frac{b}{a}}^{\infty} g_X^{2-\alpha}(x)dx\right].
\end{eqnarray*} 
Hence the proof.
\end{proof}

\section{Weighted Mathai-Haubold Residual Entropy}\label{S3}
In this section, we propose weighted Mathai-Haubold residual entropy (WMHRE) for a non-negative rv $X$. Consider the following definition.
\begin{de}
	For a non-negative continuous rv $X$ with pdf $g(x)$, the WMHRE is defined as
	\begin{eqnarray}\label{eq8}
	M^w_{\alpha}(g;t)=\frac{1}{\alpha-1}\left[\frac{\int_{t}^{\infty}x g^{2-\alpha}(x)dx}{\bar{G}^{2-\alpha}(t)}-1\right], \;\alpha\neq1,0<\alpha<2.
	\end{eqnarray}
\end{de}

We calculate WMHRE for some distributions and present them in Table \ref{t2}.

\begin{table}[H]
	\centering
	\caption{WMHPE for some distributions.}\label{t2}
	\begin{tabular}{|c|c|c|}
		\hline
		\textbf{Distributions} &\textbf{$g(x)$} & \textbf{ $M^w_{\alpha}(g;t)$}\\
		\hline
		
		Uniform & $\frac{1}{b-a}; a<x<b$ & $\frac{1}{\alpha-1}\left[\frac{(b+t)(b-t)^{\alpha-1}}{2}-1\right]$\\
		
		Exponential & $\lambda e^{-\lambda x}; x>0, \lambda >0 $ & $\frac{1}{\alpha-1}\left[\frac{t\lambda^{1-\alpha}}{2-\alpha}+\frac{\lambda^{-\alpha}}{(2-\alpha)^2}-1\right]$\\
		
		Power & $\theta x^{\theta-1}; 0<x<1, \theta>0$ & $\frac{1}{\alpha-1}\left[\frac{\theta^{2-\alpha}\{1-t^{(\theta-1)(2-\alpha)+2}\}}{\{(\theta-1)(2-\alpha)+2\}(1-t^{\theta})^{2-\alpha}}-1\right]$ \\
		\hline
	\end{tabular}
\end{table}
\pagebreak

In the following theorem we show $M^w_{\alpha}(g;t)$ uniquely determines the underlying distribution.

\begin{thm}
	Let $X$ be the non-negative absolutely continuous rv with pdf $g(x)$ and survival function $\bar{G}(x)$. Assume that $M^w_{\alpha}(g;t)<\infty,\; \alpha\neq 1,\;0<\alpha<2$ and decreasing, then $M^w_{\alpha}(g;t)$ uniquely determines $\bar{G}(x)$.
\end{thm}
\begin{proof}
	From equation (\ref{eq8}) we have
\begin{eqnarray*}
\int_t^{\infty} x g^{2-\alpha}(x)dx = \bar{G}^{2-\alpha}(t)\left[(\alpha-1)M^w_{\alpha}(g;t)+1\right].
\end{eqnarray*}
Differentiating both sides with respect to $t$,
\begin{equation*}
-tg^{2-\alpha}(t) = (\alpha-1)\bar{G}^{2-\alpha}(t)\frac{d}{dt}M^w_{\alpha}(g;t)-(2-\alpha)\bar{G}^{1-\alpha}(t)g(t)\left[(\alpha-1)M^w_{\alpha}(g;t)+1\right].
\end{equation*}
This implies
\begin{equation}{\label{eq9}}
(\alpha-1)\frac{d}{dt}M^w_{\alpha}(g;t)-(2-\alpha)h(t)\left[(\alpha-1)M^w_{\alpha}(g;t)+1\right]+t h^{2-\alpha}(t)=0,
\end{equation}
where $h(t)=\frac{g(t)}{\bar{G}(t)}$ is the hazard rate function. Let $h(t)$ be a positive solution of the equation $p(h)=0$, where
\begin{equation}{\label{eq10}}
p(h)=(\alpha-1)\frac{d}{dt}M^w_{\alpha}(g;t)-(2-\alpha)h\left[(\alpha-1)M^w_{\alpha}(g;t)+1\right]+ th^{2-\alpha}.
\end{equation}
Furthermore,
\begin{equation*}
p'(h)=-(2-\alpha)\left[(\alpha-1)M^w_{\alpha}(g;t)+1\right]+(2-\alpha)t h^{1-\alpha},
\end{equation*}
so that $p'(h)=0$ if and only if 
\begin{equation*}
h=\left[t^{-1} \{(\alpha-1)M^w_{\alpha}(g;t)+1\}\right]^{\frac{1}{1-\alpha}}.
\end{equation*}
Also
\begin{equation*}
p''(h)=(\alpha-1)(2-\alpha)t h^{-\alpha}.
\end{equation*}
Therefore, assuming $\frac{d}{dt}M^w_{\alpha}(g;t)>0$, two cases arises.

\underline{Case I:} When $0<\alpha<1$, then $p''(h)>0$. Therefore $p(h)$ is strictly convex. Now, since $\alpha-1<0$, $p(0)<0$. And on the other hand, $\lim_{x \to \infty} p(h)=+\infty$ because $2-\alpha>1$. Hence the function $p(h)$ is first increasing and then decreasing which implies $p(h)$ has one root and a unique minimum point at $\left[t^{-1} \{(\alpha-1)M^w_{\alpha}(g;t)+1\}\right]^{\frac{1}{1-\alpha}}$.

\underline{Case II:} When $1<\alpha <2$, then $p''(h)<0$. Thus $p(h)$ is strictly concave. Again, since $\alpha-1>0$, $p(0)>0$. Also because, $0<2-\alpha<1, \; \lim_{x \to \infty} p(h)=-\infty$. Thus the function $p(h)$ is first decreasing and then increasing and has one root and a unique maximum point at $\left[t^{-1} \{(\alpha-1)M^w_{\alpha}(g;t)+1\}\right]^{\frac{1}{1-\alpha}}$.

Now by combining both the cases, we conclude that $M^w_{\alpha}(g;t)$ uniquely determines the hazard function, which uniquely determines the survival function and hence the proof.
\end{proof}


The next theorem gives a characterization result for Weibull distrubution using WMHRE measure.

\begin{thm}
    If $M^w_{\alpha}(g;t)=k \; \forall t$, and $1<\alpha<2$ where $k$ is constant, then $X$ follows Weibull distribution.
\end{thm}
\begin{proof}
    Let $M^w_{\alpha}(g;t)=k$. Then $\frac{d}{dt}M^w_{\alpha}(g;t)=0$. Thus from (\ref{eq9}), we have
    \begin{eqnarray*}
        h(t)&=& [t^{-1}(2-\alpha)\{(\alpha-1)k+1\}]^{1/{1-\alpha}}\\
            &=& ct^{\frac{1}{{\alpha-1}}},\;\text{where} \; c=(2-\alpha)\{(\alpha-1)k+1\}.
    \end{eqnarray*}
    Hence $h(t)$ follows Weibull hazard.
\end{proof}

Next, we establish effect of linear transformation on WMHRE measure. Consider the following lemma.

\begin{lem}
	Consider a linear transformation $Y=aX+b,\;a>0,b\geq 0$. Then 
	\begin{eqnarray*}
		(\alpha-1)M^w_{\alpha}(Y;t)+1=\frac{a^{\alpha}\int_{\frac{t-b}{a}}^{\infty}x g_X^{2-\alpha}(x)dx+ a^{\alpha-1}b \int_{\frac{t-b}{a}}^{\infty} g_X^{2-\alpha}(x)dx}{\bar{G}_X^{2-\alpha}\left(\frac{t-b}{a}\right)}, \;\;\alpha\neq1,0<\alpha<2.
	\end{eqnarray*}
\end{lem}
\begin{proof}
	Using the density function transformation $g_Y^{2-\alpha}(y)=a^{\alpha-2} g_X^{2-\alpha}\left(\frac{y-b}{a}\right)$ and survival function transformation $\bar{G}_Y^{2-\alpha}(t)=\bar{G}_X^{2-\alpha}\left(\frac{t-b}{a}\right)$ the result follows.
\end{proof}

\begin{de}
	The ditribution function $G(x)$ has non-decreasing (non-increasing) weighted Mathai-Haubold residual entropy of order $\alpha$ IWMHRE($\alpha$) (or DWMHRE($\alpha$)) if $M^w_{\alpha}(g;t)$ is non-decreasing (non-increasing) in $t$, $t>0$. This implies that $G(x)$ has IWMHRE($\alpha$) (or DWMHRE($\alpha$)) if 
	\begin{equation*}
	\frac{d}{dt} M^w_{\alpha}(g;t) \geq (\leq) 0.
	\end{equation*}
\end{de}

\begin{thm}\label{thm3.2} For a non-negative continuous rv $X$ with WMHRE $M^w_{\alpha}(g;t)$, the following relations hold:
	\begin{enumerate}[(a)]
		\item If rv $x$ is IWMHRE($\alpha$), then
		\begin{equation*}
		h(t) \leq \left[(\alpha-2)t^{-1}\{(\alpha-1)M^w_{\alpha}(g;t)+1\}\right]^{\frac{1}{1-\alpha}}.
		\end{equation*}	
		
		\item If rv $x$ is DWMHRE($\alpha$), then
		\begin{equation*}
		h(t) \geq \left[(\alpha-2)t^{-1}\{(\alpha-1)M^w_{\alpha}(g;t)+1\}\right]^{\frac{1}{1-\alpha}}.
		\end{equation*}
	\end{enumerate}
\end{thm}
\begin{proof}
	If $X$ is IWMHRE($\alpha$), then $\frac{d}{dt} \bar{M}^w_{\alpha}(g;t) \geq (\leq) 0$. We have
	\begin{alignat*}{4}
	&&(\alpha-1)\frac{d}{dt}M^w_{\alpha}(g;t)-(2-\alpha)h(t)\left[(\alpha-1)\bar{M}^w_{\alpha}(g;t)+1\right] &&\quad&=& -t h^{2-\alpha}(t) \\
	\Rightarrow && -(2-\alpha)h(t)\left[(\alpha-1)M^w_{\alpha}(g;t)+1\right] &&\quad&\geq & -t h^{2-\alpha}(t).
	\end{alignat*}
	\begin{equation*}
	\therefore h(t) \leq \left[(\alpha-2)t^{-1}\{(\alpha-1)M^w_{\alpha}(g;t)+1\}\right]^{\frac{1}{1-\alpha}}.
	\end{equation*}
	The proof of part (b) is similar as part (a).
\end{proof}

\section{Weighted Mathai-Haubold Past Entropy}\label{S4}
In this section, we propose weighted Mathai-Haubold past entropy (WMHRE) for a non-negative rv $X$. 
\begin{de}
	For a non-negative continuous rv $X$ with pdf $g(x)$, the WMHRE is defined as
	\begin{eqnarray}\label{eq10}
	\bar{M}^w_{\alpha}(g;t)=\frac{1}{\alpha-1}\left[\frac{\int_{0}^{t}x g^{2-\alpha}(x)dx}{G^{2-\alpha}(t)}-1\right], \;\alpha\neq1,0<\alpha<2.
	\end{eqnarray}
\end{de}

WMHPE for some continuous distributions are provided in Table \ref{t3}.

\begin{table}[H]
	\centering
	\caption{WMHPE for some distributions.}\label{t3}
	\begin{tabular}{|c|c|c|}
		\hline
		\textbf{Distributions} &\textbf{$g(x)$} & \textbf{ $\bar{M}^w_{\alpha}(g;t)$}\\
		\hline
		
		Uniform & $\frac{1}{b-a}; a<x<b$ & $\frac{1}{\alpha-1}\left[\frac{(t+a)(t-a)^{\alpha-1}}{2}-1\right]$\\

        Exponential & $g(x)=\lambda e^{-\lambda x}; x>0, \lambda>0$ & $\frac{1}{\alpha-1}\left[ \frac{1-e^{-(2-\alpha)t}\{1+(2-\alpha)t\}}{(2-\alpha)^2(1-e^{-t})^{2-\alpha}}-1\right]$ \\
		
		Power & $\theta x^{\theta-1}; 0<x<1, \theta>0$ & $\frac{1}{\alpha-1}\left[\frac{\theta^{2-\alpha}t^{\alpha}}{(\theta-1)(2-\alpha)+2}-1\right]$ \\
		
		\hline
	\end{tabular}
\end{table}

In the following theorem we show $\bar{M}^w_{\alpha}(g;t)$ uniquely determines the following distribution.

\begin{thm}
	Let $X$ be the non-negative rv having continuous distribution with pdf $g(x)$ and cdf $G(x)$. Assume that $\bar{M}^w_{\alpha}(g;t)<\infty,\; \alpha\neq 1,\;0<\alpha<2$ and decreasing, then $\bar{M}^w_{\alpha}(g;t)$ uniquely determines $G(x)$.
\end{thm}
\begin{proof}
	From equation (\ref{eq9}) we have
	\begin{eqnarray*}
		\int_0^{t} x g^{2-\alpha}(x)dx = G^{2-\alpha}(t)\left[(\alpha-1)\bar{M}^w_{\alpha}(g;t)+1\right].
	\end{eqnarray*}
	Differentiating both sides with respect to $t$,
	\begin{equation*}
	(\alpha-1)\frac{d}{dt}\bar{M}^w_{\alpha}(g;t)+(2-\alpha)r(t)\left[(\alpha-1)\bar{M}^w_{\alpha}(g;t)+1\right]-t r^{2-\alpha}(t)=0,
	\end{equation*}
	where $r(t)=\frac{g(t)}{G(t)}$ is the reversed hazard function. Let $r(t)$ be a positive solution of the equation $q(r)=0$, where
	\begin{equation}{\label{eq11}}
	q(r)=(\alpha-1)\frac{d}{dt}\bar{M}^w_{\alpha}(g;t)-(2-\alpha)r\left[(\alpha-1)\bar{M}^w_{\alpha}(g;t)+1\right]-t r^{2-\alpha}.
	\end{equation}
    Differentiating with respect to $r$,
	\begin{equation*}
	q'(r)=-(2-\alpha)\left[(\alpha-1)\bar{M}^w_{\alpha}(g;t)+1\right]-(2-\alpha)t r^{1-\alpha},
	\end{equation*}
	so that $q'(r)=0$ if and only if 
	\begin{equation*}
	r=\left[t^{-1} \{(\alpha-1)\bar{M}^w_{\alpha}(g;t)+1\}\right]^{\frac{1}{1-\alpha}}.
	\end{equation*}
	Also
	\begin{equation*}
	q''(r)=(\alpha-1)(2-\alpha)t r^{-\alpha}.
	\end{equation*}
	Therefore, assuming $\frac{d}{dt}M^w_{\alpha}(g;t)>0$, two cases arises.
	
	\underline{Case I:} Let $0<\alpha<1$, then $q''(r)>0$. Thus $q(r)$ is first decreasing and then increasing, which implies the function $q(r)$ has one root and a unique minimum point at $\left[t^{-1} \{(\alpha-1)\bar{M}^w_{\alpha}(g;t)+1\}\right]^{\frac{1}{1-\alpha}}$.
	
	\underline{Case II:} Let $1<\alpha <2$, then $q''(r)<0$. Thus $q(r)$ is first increasing and then decreasing, which implies the function $q(r)$ has one root and a unique maximum point at $\left[t^{-1} \{(\alpha-1)\bar{M}^w_{\alpha}(g;t)+1\}\right]^{\frac{1}{1-\alpha}}$.
	
	Now by combining both the cases, we conclude that $\bar{M}^w_{\alpha}(g;t)$ uniquely determines the reversed hazard function, which uniquely determines the distribution function and hence the proof.
\end{proof}


\begin{lem}
	Consider a linear transformation $Y=aX+b,\;a>0,b\geq 0$. Then 
	\begin{eqnarray*}
		(\alpha-1)\bar{M}^w_{\alpha}(Y;t)+1=\frac{a^{\alpha}\int_{-\frac{b}{a}}^{\frac{t-b}{a}}x g_X^{2-\alpha}(x)dx+ a^{\alpha-1}b \int_{-\frac{b}{a}}^{\frac{t-b}{a}} g_X^{2-\alpha}(x)dx}{G_X^{2-\alpha}\left(\frac{t-b}{a}\right)}, \;\;\alpha\neq1,0<\alpha<2.
	\end{eqnarray*}
\end{lem}
\begin{proof}
	The result follows by the density transformation $g_Y^{2-\alpha}(y)=a^{\alpha-2} g_X^{2-\alpha}\left(\frac{y-b}{a}\right)$ and cumulative density transformation $G_Y^{2-\alpha}(t)=G_X^{2-\alpha}\left(\frac{t-b}{a}\right)$.
\end{proof}

\begin{de}
	The ditribution function $G(x)$ has non-decreasing (non-increasing) weighted Mathai-Haubold past entropy of order $\alpha$ IWMHPE($\alpha$) (or DWMHPE($\alpha$)) if $\bar{M}^w_{\alpha}(g;t)$ is non-decreasing (non-increasing) in $t$, $t>0$. This implies that $G(x)$ has IWMHPE($\alpha$) (or DWMHPE($\alpha$)) if 
	\begin{equation*}
	\frac{d}{dt} \bar{M}^w_{\alpha}(g;t) \geq (\leq) 0.
	\end{equation*}
\end{de}

\begin{thm} The following relations hold for a non-negative continuous rv $X$ with WMHPE $\bar{M}^w_{\alpha}(g;t)$:
	\begin{enumerate}[(a)]
		\item If rv $x$ is IWMHPE($\alpha$), then
		\begin{equation*}
		r(t) \geq \left[(2-\alpha)t^{-1}\{(\alpha-1)\bar{M}^w_{\alpha}(g;t)+1\}\right]^{\frac{1}{1-\alpha}}.
		\end{equation*}	
		
		\item If rv $x$ is DWMHPE($\alpha$), then
		\begin{equation*}
		r(t) \leq \left[(2-\alpha)t^{-1}\{(\alpha-1)\bar{M}^w_{\alpha}(g;t)+1\}\right]^{\frac{1}{1-\alpha}}.
			\end{equation*}
	\end{enumerate}
\end{thm}
\begin{proof}
The proof is similar to Theorem-\ref{thm3.2}.
\end{proof}

Now, we establish a relationship between the three weighted Mathai-Haubold entropy measures.

We know,
\begin{equation}\label{eq12}
\int_0^{\infty}x g^{2-\alpha}(x)dx = \int_0^t x g^{2-\alpha}(x)dx + \int_t^{\infty}x g^{2-\alpha}(x)dx.
\end{equation}

Now, from the definition of WMHE, WMHRE and WMHPE defined in equation (\ref{eq7}), (\ref{eq8}) and (\ref{eq10}) respectively we can write
\begin{equation*}
\int_0^{\infty}x g^{2-\alpha}(x)dx = (\alpha-1)M^w_{\alpha}(X)+1,
\end{equation*}
\begin{equation*}
\int_0^t x g^{2-\alpha}(x)dx = G^{2-\alpha}(t)\left[(\alpha-1)\bar{M}^w_{\alpha}(g;t)+1\right],
\end{equation*}
and
\begin{equation*}
\int_t^{\infty} x g^{2-\alpha}(x)dx = \bar{G}^{2-\alpha}(t)\left[(\alpha-1)M^w_{\alpha}(g;t)+1\right].
\end{equation*}

Hence substuting the relations in equation (\ref{eq12}) we obtain the relation between WMHE, WMHRE and WMHPE as
\begin{equation}\label{eq13}
M^w_{\alpha}(X) = \bar{G}^{2-\alpha}(t)M^w_{\alpha}(g;t) + G^{2-\alpha}(t)\bar{M}^w_{\alpha}(g;t) + \frac{\bar{G}^{2-\alpha}(t)+G^{2-\alpha}(t)-1}{\alpha-1}.
\end{equation} 

\begin{re}
	If $\alpha\to 1$ then $\bar{G}^{2-\alpha}(t) \to \bar{G}(t)$ and $G^{2-\alpha}(t) \to G(t)$. Then the relation (\ref{eq13}) becomes
	\begin{equation*}
	M^w(X) \approx \bar{G}(t)M^w(g;t) + G(t)\bar{M}^w(g;t).
	\end{equation*}
\end{re}

In the next section, we provide some bounds for the proposed measures.

\section{Some inequalities and bounds}\label{S5}

 The following theorem gives lower bounds for WMHE, WMHRE and WMHPE in terms of Shannon entropy.

\begin{thm}
	For a continuous rv $X$ with pdf $g(x)$, cdf $G(x)$ and survival function $\bar{G}(x)$ the following inequalities hold:
	\begin{enumerate}[(i)]
		\item $M^w_{\alpha}(X) \geq \frac{1}{\alpha-1}\left[\exp\{(\alpha-1)H(X)+E(\log X)\}-1\right]$,
		
		\item $M^w_{\alpha}(g;t) \geq \frac{1}{\alpha-1}\left[\exp\{(\alpha-1)H(g;t)+E(\log X|X>t)\}-1\right]$,
		
		\item $\bar{M}^w_{\alpha}(g;t) \geq \frac{1}{\alpha-1}\left[\exp\{(\alpha-1)\bar{H}(g;t)+E(\log X|X\leq t)\}-1\right]$.
	\end{enumerate}
\end{thm}
\begin{proof}
	Using log-sum inequality, we have
	\begin{eqnarray}
	\int_0^{\infty} g(x) \log \left[{\frac{g(x)}{x g^{2-\alpha}(x)}}\right]dx &\geq& \int_0^{\infty} g(x) \times \log \left[{\frac{\int_0^{\infty}g(x)dx}{\int_0^{\infty}x g^{2-\alpha}(x)}dx}\right] \nonumber \\
			&=& \log\left[\frac{1}{(\alpha-1)M^w_{\alpha}(X)+1}\right] \nonumber \\
			&=& -\log\left[(\alpha-1)M^w_{\alpha}(X)+1\right].\label{eq14}
	\end{eqnarray}
	After some simplification, the LHS of (\ref{eq14}) reduces to $-(\alpha-1)H(X)-E(\log X)$ from which the result follows. The proof of (ii) and (iii) are similar to (i). 
\end{proof}

Next we show a relationship between WMHE, WMHRE and WMHPE using log-sum inequality.
\begin{lem}
	For a non-negative continuous rv $X$ with density function $g(x)$ and survival function $\bar{G}(x)$, 
	\begin{equation*}
	L_o = \frac{L_1^{G(t)} \times L_2^{\bar{G}(t)}}{\bar{G}^{(\alpha-1)\bar{G}(t)}(t) \times G^{(\alpha-1)G(t)}(t)}.
	\end{equation*} 
	where
	\begin{eqnarray*}
	L_0 &=& \exp{[(\alpha-1)H(X)+E(\log{X})]},\\
	L_1 &=& \exp{[(\alpha-1)H(g;t)+E(\log{X}|X> t)]}\\
	L_2 &=& \exp{[(\alpha-1)\bar{H}(g;t)+E(\log{X}|X\leq t)]}.
	\end{eqnarray*}
\end{lem}
\begin{proof}
	We know the log-moment is a probability-weighted mixture of residual and past log-moments and can be written as,
	\begin{equation}\label{eq15}
	E(\log X)= G(t)E(\log X|X> t) + \bar{G}(t)E(\log X|X\leq t)
	\end{equation}
	and Shannon entropy as
	\begin{equation}\label{eq16}
	H(X) = \bar{G}(t)H(g;t) + G(t)\bar{H}(g;t) -G(t)\log{G(t)} - \bar{G}(t)\log{\bar{G}(t)}.
	\end{equation}
	Now using (\ref{eq15}) and (\ref{eq16}) we can write
	\begin{eqnarray}
	&&(\alpha-1)H(X)+E(\log X) \nonumber\\
	= &&\bar{G}(t)\left[(\alpha-1)H(g;t)+E(\log{X}|X> t)\right] + \nonumber G(t)\left[(\alpha-1)\bar{H}(g;t)+E(\log{X}|X\leq t)\right] \nonumber \\
	&&- (\alpha-1)\left[\bar{G}(t)\log{\bar{G}(t)}+G(t)\log{G(t)}\right].\label{eq17}
	\end{eqnarray}
Taking exponential on both sides of equation (\ref{eq17}) the result follows.
\end{proof} 

\begin{re}
	$M^w_{\alpha}(X)$, lower bound $=\frac{1}{\alpha-1}[L_0-1]$ is determined by a geometric mixture of the residual and past lower bounds.
\end{re}

\begin{ex}
	Let $X\sim U(a,b)$ then $G(t)=\frac{t}{b}$, $\bar{G}(t)=\frac{b-t}{b}$ and
	\begin{eqnarray*}
	(\alpha-1)H(X)+E(\log X) &=& \alpha \log b-1,\\
	(\alpha-1)H(g;t)+E(\log{X}|X> t) &=& (\alpha-1)\log(b-t)+\frac{b\log b-t\log t-(b-t)}{b-t},\\
	(\alpha-1)\bar{H}(g;t)+E(\log{X}|X\leq t) &=& \alpha \log t-1.
	\end{eqnarray*}
So that
	\begin{eqnarray*}
		L_0 &=& b^{\alpha} e^{-1},\\
		L_1 &=& \exp\left[(\alpha-1)\log(b-t)+\frac{b\log b-t\log t-(b-t)}{b-t}\right],\\
		L_2 &=& t^{\alpha} e^{-1}.
	\end{eqnarray*}	
Hence
\begin{equation*}
L_0 = \frac{L_1^{\frac{b-t}{b}} \times L_2^{\frac{t}{b}}}{\left[\frac{b-t}{b}\right]^{(\alpha-1)\frac{b-t}{b}} \times \left[\frac{t}{b}\right]^{(\alpha-1)\frac{t}{b}}}.
\end{equation*}
\end{ex}

Next we introduce Hardy's inequality which will be used to obtain various results.

If $f$ be a non-negative measurable function, then an integral version of Hardy's inequality states that
\begin{eqnarray}\label{eq18}
\int_0^{\infty} \left(\frac{1}{x}\int_0^x f(u)du\right)^p \leq \left(\frac{p}{p-1}\right)^p \int_0^{\infty} f^p(x) dx
\end{eqnarray}
where $p>1$. Equality holds if and only if $f(x)=0$ almost everywhere, if the RHS is finite. For further details see Hardy \cite{11}.

\begin{thm}
	For a non-negative continuous rv $X$ with pdf $g(x)$, the following inequalities hold:
	\begin{enumerate}[(i)]
		\item $M^w_{\alpha}(X) \geq \frac{1}{\alpha-1} \left[\left(\frac{1-\alpha}{2-\alpha}\right)^{2-\alpha} \int_0^{\infty} \left(\frac{1}{x} \int_0^x u^{\frac{1}{2-\alpha}}g(u)du\right)^{2-\alpha}dx -1\right]$,
		
		\item $M^w_{\alpha}(g;t) \geq \frac{1}{\alpha-1} \left[\frac{\left(\frac{1-\alpha}{2-\alpha}\right)^{2-\alpha} \int_t^{\infty} \left(\frac{1}{x-t} \int_t^x u^{\frac{1}{2-\alpha}}g(u)du\right)^{2-\alpha}dx}{\bar{G}^{2-\alpha}(t)}-1\right]$,
		
		\item $\bar{M}^w_{\alpha}(g;t) \geq \frac{1}{\alpha-1} \left[\frac{\left(\frac{1-\alpha}{2-\alpha}\right)^{2-\alpha} \int_0^{t} \left(\frac{1}{x} \int_0^x u^{\frac{1}{2-\alpha}}g(u)du\right)^{2-\alpha}dx}{G^{2-\alpha}(t)}-1\right]$.
	\end{enumerate}
\end{thm}	
\begin{proof}
	\begin{enumerate}[(i)]
		\item 
	By choosing the function $f(x)=x^{\frac{1}{2-\alpha}} g(x)$ we get,
	\begin{equation*}
	\int_0^{\infty} f^{2-\alpha}(x)dx = \int_0^{\infty} x g^{2-\alpha}(x)dx.
	\end{equation*}
	So Hardy's inequality can be applied with $p=2-\alpha>1$ (i.e., $\alpha<1$). Hence
	\begin{eqnarray}\label{eq19}
	\int_0^{\infty} \left(\frac{1}{x}\int_0^x f(u)du\right)^{2-\alpha}dx \leq \left(\frac{2-\alpha}{1-\alpha}\right)^{2-\alpha} \int_0^{\infty} f^{2-\alpha}(x) dx
	\end{eqnarray} 
	Now substituting $f$ in equation (\ref{eq19}), we get
	\begin{alignat*}{4}
	 	&& \int_0^{\infty} \left(\frac{1}{x}\int_0^x u^{\frac{1}{2-\alpha}}g(u)du\right)^{2-\alpha}dx &&\quad& \leq && \left(\frac{2-\alpha}{1-\alpha}\right)^{2-\alpha} \int_0^{\infty} xg^{2-\alpha}(x) dx \\
	 \Rightarrow && \left(\frac{1-\alpha}{2-\alpha}\right)^{2-\alpha} \int_0^{\infty} \left(\frac{1}{x}\int_0^x u^{\frac{1}{2-\alpha}}g(u)du\right)^{2-\alpha}dx &&\quad& \leq && (\alpha-1)M^w_{\alpha}(X)+1.
	\end{alignat*}
	\begin{equation*}
	\therefore M^w_{\alpha}(X) \geq \frac{1}{\alpha-1} \left[\left(\frac{1-\alpha}{2-\alpha}\right)^{2-\alpha} \int_0^{\infty} \left(\frac{1}{x} \int_0^x u^{\frac{1}{2-\alpha}}g(u)du\right)^{2-\alpha}dx -1\right]. 
	\end{equation*}
	
\end{enumerate}

For (ii), first we shift the integral. Let $y=x-t \Rightarrow x=y+t, \; y\geq 0$. Then
\begin{equation*}
\int_t^{\infty} xg^{2-\alpha}(x)dx = \int_0^{\infty} (y+t) g^{2-\alpha}(y+t)dy.
\end{equation*}
Let us define $\phi(y)=(y+t)^{\frac{1}{2-\alpha}}g(y+t)$. Then
\begin{equation*}
\int_0^{\infty} {\phi}^{2-\alpha}(y)dy = \int_0^{\infty} (y+t) g^{2-\alpha}(y+t)dy.
\end{equation*}
For $p=2-\alpha>1$ (i.e., $\alpha<1$), Hardy's inequality gives,
\begin{eqnarray}\label{eq20}
\int_0^{\infty} \left(\frac{1}{y}\int_0^y \phi(u)du\right)^{2-\alpha}dy \leq \left(\frac{2-\alpha}{1-\alpha}\right)^{2-\alpha} \int_0^{\infty} \phi^{2-\alpha}(y) dy.
\end{eqnarray} 
Hence by substituting $\phi$ in equation (\ref{eq20}) we get
\begin{equation}\label{eq21}
M^w_{\alpha}(g;t) \geq \frac{1}{\alpha-1} \left[\frac{\left(\frac{1-\alpha}{2-\alpha}\right)^{2-\alpha} \int_0^{\infty} \left(\frac{1}{y} \int_0^y (u+t)^{\frac{1}{2-\alpha}}g(u+t)du\right)^{2-\alpha}dy}{\bar{G}^{2-\alpha}(t)}-1\right]
\end{equation}
By replacing $y=x-t$ in (\ref{eq21}) we obtain the result. The proof of part (iii) is similar to part (i).
\end{proof}
	
Next we provide single unified Hardy-type inequality for WMHE, WMHRE and WMHPE.

\begin{lem}
	For any interval $I=[A,B]\subseteq [0,\infty]$, the unified Hardy inequality can be defined as,
	\begin{eqnarray}\label{eq22}
	\mathcal{M}^w_{\alpha}(I) \geq \frac{1}{\alpha-1} \left[\frac{\left(\frac{1-\alpha}{2-\alpha}\right)^{2-\alpha} \int_A^B \left(\frac{1}{x-A} \int_A^x u^{\frac{1}{2-\alpha}}g(u)du\right)^{2-\alpha}dx}{P^{2-\alpha}(I)}-1\right]
	\end{eqnarray}
where $P(I)=\int_A^B g(x)dx$. Then
\begin{enumerate}[a)]
	\item $I=[0,\infty) \Rightarrow \mathcal{M}^w_{\alpha}(I)= M^w_{\alpha}(I)$,
	\item $I=[t,\infty) \Rightarrow \mathcal{M}^w_{\alpha}(I)= M^w_{\alpha}(g;t)$,
	\item $I=[0,t] \Rightarrow \mathcal{M}^w_{\alpha}(I)= \bar{M}^w_{\alpha}(g;t)$.
\end{enumerate}	
\end{lem}
\begin{proof}
	The proof is straight forward and hence omitted.
\end{proof}

\begin{ex}
	Let the rv $X$ have a power distribution with pdf $g(x)=\theta x^{\theta-1},\; 0<x<1,\; \theta>0$. Then based on Hardy's inequality,
	
\begin{eqnarray}
M^w_{\alpha}(X) &\geq& \frac{1}{\alpha-1}\left[\left(\frac{1-\alpha}{2-\alpha}\right)^{2-\alpha} \left(\frac{\theta}{\theta + \frac{1}{2-\alpha}}\right)^{2-\alpha}\left(\frac{1}{(\theta-1)(2-\alpha)+2}\right) -1\right],\nonumber\\
M^w_{\alpha}(g;t) &\geq& \frac{1}{\alpha-1}\left[\frac{\left(\frac{1-\alpha}{2-\alpha}\right)^{2-\alpha} \left(\frac{\theta}{\theta + \frac{1}{2-\alpha}}\right)^{2-\alpha} \int_t^1\left(\frac{x^{\theta+\frac{1}{2-\alpha}}-t^{\theta+\frac{1}{2-\alpha}}}{x-t}\right)dx}{(1-t^{\theta})^{2-\alpha}} -1\right],\nonumber\\
\bar{M}^w_{\alpha}(g;t) &\geq& \frac{1}{\alpha-1}\left[\left(\frac{1-\alpha}{2-\alpha}\right)^{2-\alpha} \left(\frac{\theta}{\theta + \frac{1}{2-\alpha}}\right)^{2-\alpha}\left(\frac{t^{\alpha}}{(\theta-1)(2-\alpha)+2}\right) -1\right].\nonumber
\end{eqnarray}
\end{ex}

\section{Nonparametric Estimation}\label{S6}
In this section, we propose kernel based nonparametric estimators for weighted Mathai-Haubold entropy, weighted residual and weighted past Mathai-Haubold entropy. Kernel based estimation for entropy and related functionals have gain quite a lot of attention in recent years. See, for example, Maya \cite{13}, Maya et al. \cite{14}, Chakraborty and Pradhan \cite{15} and the references therein.  The kernel density estimator $g(x)$ given by Rosenblatt \cite{12} is defined by
	\begin{eqnarray}\label{eq23}
		g_n(x) = \frac{1}{nb_n}\sum_{i=1}^n K\left(\frac{x-X_i}{b_n}\right) 
	\end{eqnarray}
where $K(\cdot)$ is a kernel function with $\int_{-\infty}^{\infty}K(u)du=1$ and $\int_{-\infty}^{\infty}uK(u)du=0$ and $b_n$ is a sequence of real numbers satisfying the following conditions
\begin{enumerate}[(i)]
    \item $b_n \to 0$ as $n \to \infty$ and
    \item $nb_n \to \infty$ as $n \to \infty$.
\end{enumerate}

A nonparametric kernel-based estimators for $M^w_{\alpha,n}(X)$, $M^w_{\alpha,n}(g;t)$ and $\bar{M}^w_{\alpha}(g;t)$ are as follows

\begin{eqnarray}
M^w_{\alpha,n}(X) &=& \frac{1}{\alpha-1}\left[\int_{0}^{\infty}x g_n^{2-\alpha}(x)dx -1\right], \;\alpha\neq1,0<\alpha<2,\nonumber\\
M^w_{\alpha,n}(g;t) &=&\frac{1}{\alpha-1}\left[\frac{\int_{t}^{\infty}x g_n^{2-\alpha}(x)dx}{\bar{G}_n^{2-\alpha}(t)}-1\right], \;\alpha\neq1,0<\alpha<2,\nonumber\\
\bar{M}^w_{\alpha,n}(g;t) &=& \frac{1}{\alpha-1}\left[\frac{\int_{0}^{t}x g_n^{2-\alpha}(x)dx}{G_n^{2-\alpha}(t)}-1\right], \;\alpha\neq1,0<\alpha<2.\nonumber
\end{eqnarray}
where $\bar{G}_n(t)=1-G_n(t)=\int_t^{\infty}g_n(x)dx$ is the empirical survival estimator.

From the consistency of the kernel density estimator, it is easy to show that these estimators are consistent.\\

Now we extend these estimators under dependent data. Consider the following definition.
\begin{de}
Let $\{X_i,i\geq1\}$ be a sequence of rvs. For a set of positive integer $n$,
\begin{eqnarray}\label{eq8}
\rho(n)=\sup_{k\geq1}\left[|P(A\cap B)|-P(A)P(B)\right];\;\;A\in\mathscr{F}_1^k,B\in\mathscr{F}_{k+n}^{\infty},
\end{eqnarray}
where $\mathscr{F}_i^k$ denote the $\sigma$-field of events obtained by $\{X_j:i\leq j\leq k\}$. The sequence is said to be $\rho$-mixing if $\rho(n)\to 0$ as $n\to\infty$.
\end{de}
\pagebreak
Under $\rho$-mixing dependence conditions, the expressions for bias and variance of $g_n(x)$ (Masry \cite{16}), $G_n(t)$ (Maya \cite{13}) and $\bar{G}_n(t)$ (Maya et al. \cite{14}) are
 \begin{eqnarray}
Bais(g_n(x)) &\simeq& \frac{b_n^s C_s}{s!}g^{(s)}(x)\nonumber\\
Var(g_n(x)) &\simeq& \frac{1}{nb_n}g(x)C_K,\nonumber\\
Bais(G_n(t)) &\simeq& \frac{b_n^s C_s}{s!}\int_0^t g^{(s)}(x)dx, \nonumber\\
Var(G_n(t)) &\simeq& \frac{1}{nb_n}g(x)C_K\int_0^t g(x)dx,\nonumber\\
Bais(\bar{G}_n(t)) &\simeq& \frac{b_n^s C_s}{s!}\int_t^{\infty} g^{(s)}(x)dx, \;\text{and} \nonumber\\
Var(\bar{G}_n(t)) &\simeq& \frac{1}{nb_n}g(x)C_K\int_t^{\infty} g(x)dx.\nonumber
\end{eqnarray}
where $C_s=\int_{-\infty}^{\infty}u^sK(u)du$ and $C_k=\int_{-\infty}^{\infty}K^2(u)du$. For simplification we define
\begin{alignat}{10}
&& m_n^w(X) &&\quad& =&& \int_0^{\infty}xg_n^{2-\alpha}(x)dx &&,\quad&& m^w(X) &&\quad& =&& \int_0^{\infty}xg^{2-\alpha}(x)dx , \nonumber\\
&& M_n^w(g;t) &&\quad& =&&\int_{t}^{\infty}x g_n^{2-\alpha}(x)dx &&,\quad&&  M^w(g;t) &&\quad& =&&\int_{t}^{\infty}x g^{2-\alpha}(x)dx , \nonumber\\
&& \bar{M}_n^w(g;t) &&\quad& =&& \int_{0}^{t}x g_n^{2-\alpha}(x)dx &&,\quad&&  \bar{M}^w(g;t) &&\quad& =&& \int_{0}^{t}x g^{2-\alpha}(x)dx	,\nonumber 
\end{alignat}
and
\begin{alignat}{10}
&& \bar{D}_n^w(g;t) &&\quad& =&& \bar{G}_n^{2-\alpha}(t) &&,\quad&& \bar{D}^w(g;t) &&\quad& =&& \bar{G}^{2-\alpha}(t) , \nonumber\\
&& D_n^w(g;t) &&\quad& =&& G_n^{2-\alpha}(t) &&,\quad&& D^w(g;t) &&\quad& =&& G^{2-\alpha}(t) . \nonumber
\end{alignat}
Therefore,
\begin{eqnarray}
M^w_{\alpha,n}(X) &=& \frac{1}{\alpha-1}\left[m_n^w(X)-1\right],\label{eq24}\\
M^w_{\alpha,n}(g;t) &=&\frac{1}{\alpha-1}\left[\frac{M_n^w(g;t)}{\bar{D}_n^w(g;t)}-1\right],\label{eq25}\\
\bar{M}^w_{\alpha,n}(g;t) &=& \frac{1}{\alpha-1}\left[\frac{\bar{M}_n^w(g;t)}{D_n^w(g;t)}-1\right], \label{eq26}
\end{eqnarray}
and
\begin{eqnarray}
M^w_{\alpha}(X) &=& \frac{1}{\alpha-1}\left[m^w(X)-1\right],\\
M^w_{\alpha}(g;t) &=&\frac{1}{\alpha-1}\left[\frac{M^w(g;t)}{\bar{D}^w(g;t)}-1\right],\\
\bar{M}^w_{\alpha}(g;t) &=& \frac{1}{\alpha-1}\left[\frac{\bar{M}^w(g;t)}{D^w(g;t)}-1\right].
\end{eqnarray}

The following theorem addresses the consistency of the estimators defined in equations (\ref{eq24}), (\ref{eq25}) and (\ref{eq26}), respectively.

\begin{thm}
    The nonparametric kernel estimators $M^w_{\alpha,n}(X)$, $M^w_{\alpha,n}(g;t)$ and $\bar{M}^w_{\alpha,n}(g;t)$ are consistent estimators of $M^w_{\alpha}(X)$, $M^w_{\alpha}(g;t)$ and $\bar{M}^w_{\alpha}(g;t)$ respectively.
\end{thm}
\begin{proof}
    By using Taylor's series expansion, we get
\begin{eqnarray*}
   \int xg_n^{2-\alpha}(x)dx \simeq \int xg^{2-\alpha}(x)dx + (2-\alpha)\int x(g_n(x)-g(x))g^{1-\alpha}(x)dx + O(x^3).   
\end{eqnarray*}
The expression for the bias and variance of $m_n(X)$ are
\begin{eqnarray}
    Bais(m_n(X)) &\simeq& \frac{(2-\alpha)b_n^sC_s}{s!} \int_0^{\infty} xg^{(s)}(x)g^{1-\alpha}(x)dx, \label{30}
\end{eqnarray}
and 
\begin{eqnarray}
    Var(m_n(X))  &\simeq& \frac{(2-\alpha)^2C_K}{nb_n} \int_0^{\infty} x^2g^{3-2\alpha}(x)dx, \label{31}
\end{eqnarray}
The mean-squared error (MSE) of $m_n(X)$ is given by
\begin{eqnarray}
    MSE(m_n(X))  &\simeq& \left[\frac{(2-\alpha)b_n^sC_s}{s!} \int_0^{\infty} xg^{(s)}(x)g^{1-\alpha}(x)dx\right]^2 \nonumber\\
    &&+ \frac{(2-\alpha)^2C_K}{nb_n} \int_0^{\infty} x^2g^{3-2\alpha}(x)dx, \label{eq32}
\end{eqnarray}
From (\ref{eq32}), as $n \to \infty$, $MSE(m_n(X)) \to 0$. Therefore, $m_n(X) \xrightarrow{P} m(X)$. Hence
\begin{eqnarray*}
    M^w_{\alpha,n}(X)=\frac{1}{\alpha-1}\left[m_n^w(X)-1\right] \xrightarrow{P} \frac{1}{\alpha-1}\left[m^w(X)-1\right]=M^w_{\alpha}(X).
\end{eqnarray*}
That is, the estimator $M^w_{\alpha,n}(X)$ is consistent estimator (in probability) to $M^w_{\alpha}(X)$.
\end{proof}

To prove the consistency of $M^w_{\alpha,n}(g;t)$, by using Taylor's expansion, the expression for the bias and variance of $M^w_{\alpha,n}(g;t)$ and $\bar{D}_n^w(g;t)$ are
\begin{eqnarray}
    Bais(M_n(g;t)) &\simeq& \frac{(2-\alpha)b_n^sC_s}{s!} \int_t^{\infty} xg^{(s)}(x)g^{1-\alpha}(x)dx, \\
    Var(M_n(g;t))  &\simeq& \frac{(2-\alpha)^2C_K}{nb_n} \int_t^{\infty} x^2g^{3-2\alpha}(x)dx, 
\end{eqnarray} 
and
\begin{eqnarray}
    Bais(\bar{D}_n^w(g;t)) &\simeq& \frac{(2-\alpha)b_n^sC_s}{s!} \bar{G}^{1-\alpha}(t) \int_t^{\infty} g^{(s)}(x)dx,  \\
    Var(\bar{D}_n^w(g;t))  &\simeq& \frac{(2-\alpha)^2 C_K}{nb_n} \bar{G}^{3-2\alpha}(t).
\end{eqnarray} 
The mean-squared error (MSE) of $M_n(g;t)$ and $\bar{D}_n^w(g;t)$ are given by
\begin{eqnarray}
    MSE(M_n(g;t))  &\simeq& \left[\frac{(2-\alpha)b_n^sC_s}{s!} \int_t^{\infty} xg^{(s)}(x)g^{1-\alpha}(x)dx\right]^2 \nonumber\\
    &&+ \frac{(2-\alpha)^2C_K}{nb_n} \int_t^{\infty} x^2g^{3-2\alpha}(x)dx, \label{eq37}
\end{eqnarray}
and 
\begin{eqnarray}
    MSE(\bar{D}_n^w(g;t))  &\simeq& \left[\frac{(2-\alpha)b_n^sC_s}{s!} \bar{G}^{1-\alpha}(t) \int_t^{\infty} g^{(s)}(x)dx\right]^2 \nonumber\\
    &&+ \frac{(2-\alpha)^2 C_K}{nb_n} \bar{G}^{3-2\alpha}(t). \label{eq38}
\end{eqnarray}

\section{Simulation Study}\label{S7}

A Monte-Carlo simulation study is conducted to assess the performance of the proposed non-parametric estimators in terms of mean square error (MSE). For reference distributions, we consider standard uniform and standard exponential distributions. We generate 1000 samples and calculate the Bias and MSE for $\alpha=0.10, \alpha=0.25, \alpha=0.50$, and $\alpha=0.75$, respectively and the results are represented in the following tables. 

\begin{table}[H]
\centering
\caption{U(0,1) : Bias and MSE of $M^w_{\alpha}(X)$}\label{t4}
\begin{tabular}{ccc c cc c cc c cc}
\hline
\textbf{n} & \multicolumn{2}{c}{$\alpha=0.10$} && \multicolumn{2}{c}{$\alpha=0.25$} && \multicolumn{2}{c}{$\alpha=0.50$} && \multicolumn{2}{c}{$\alpha=0.75$} \\
  & \multicolumn{2}{c}{True value : 0.5556} && \multicolumn{2}{c}{True value : 0.6667} && \multicolumn{2}{c}{True value : 1} && \multicolumn{2}{c}{True value : 2} \\
\cline{2-3}  \cline{5-6} \cline{8-9} \cline{11-12}
& \textbf{Bias} & \textbf{MSE} && \textbf{Bias} & \textbf{MSE} && \textbf{Bias} & \textbf{MSE} && \textbf{Bias} & \textbf{MSE} \\
\hline
 20 & 0.00058 & 0.01713  &&  0.01257 & 0.01812  && 0.01952 & 0.02878  && 0.01226 & 0.08419 \\
 30 & 0.02127 & 0.01023  &&  0.02626 & 0.01188  && 0.03245 & 0.01956 && 0.03024 & 0.06446 \\
 50 & 0.02421 & 0.00662  &&  0.02647 & 0.00726  && 0.02781 & 0.01376 && 0.04106 & 0.03881 \\
 75 & 0.02628 & 0.00448  &&  0.02949 & 0.00542  && 0.02944 & 0.00922  && 0.03766 & 0.02793 \\
100 & 0.02907 & 0.00389  &&  0.02965 & 0.00424  && 0.03199 & 0.00778  && 0.03056 & 0.02088  \\
150 & 0.02711 & 0.00254  &&  0.03043 & 0.00326  && 0.03385 & 0.00520  && 0.02925 & 0.01274 \\
200 & 0.02853 & 0.00228  &&  0.02964 & 0.00278  && 0.02983 & 0.00402  && 0.03451 & 0.01026 \\
300 & 0.02797 & 0.00177  &&  0.02851 & 0.00191  && 0.02906 & 0.00300  && 0.02873 & 0.00774 \\
\hline
\end{tabular}
\end{table}

\begin{table}[H]
\centering
\caption{Exp(1) : Bias and MSE of $M^w_{\alpha}(X)$}\label{t5}
\begin{tabular}{ccc c cc c cc c cc}
\hline
\textbf{n} & \multicolumn{2}{c}{$\alpha=0.10$} && \multicolumn{2}{c}{$\alpha=0.25$} && \multicolumn{2}{c}{$\alpha=0.50$} && \multicolumn{2}{c}{$\alpha=0.75$} \\
  & \multicolumn{2}{c}{True value : 0.8033} && \multicolumn{2}{c}{True value : 0.8980} && \multicolumn{2}{c}{True value : 1.1111} && \multicolumn{2}{c}{True value : 1.44} \\
\cline{2-3}  \cline{5-6} \cline{8-9} \cline{11-12}
& \textbf{Bias} & \textbf{MSE} && \textbf{Bias} & \textbf{MSE} && \textbf{Bias} & \textbf{MSE} && \textbf{Bias} & \textbf{MSE} \\
\hline
 20 & -0.04179 & 0.00695  &&  -0.06246 & 0.01204  && -0.11700 & 0.04097  && -0.22889 & 0.29745 \\
 30 & -0.02771 & 0.00360  &&  -0.03992 & 0.00620  && -0.08461 & 0.02419 && -0.16961 & 0.20250 \\
 50 & -0.01421 & 0.00182  &&  -0.02447 & 0.00334  && -0.05353 & 0.01251 && -0.14035 & 0.10901 \\
 75 & -0.01064 & 0.00107  &&  -0.01472 & 0.00184  && -0.03763 & 0.00777  && -0.09515 & 0.06472 \\
100 & -0.00656 & 0.00074  &&  -0.01353 & 0.00147  && -0.03048 & 0.00547  && -0.07821 & 0.05499 \\
150 & -0.00301 & 0.00047  &&  -0.00881 & 0.00086  && -0.02127 & 0.00374  && -0.05711 & 0.03340 \\
200 & -0.00161 & 0.00032  &&  -0.00438 & 0.00061  && -0.01620 & 0.00259  && -0.04809 & 0.02626 \\
300 & -0.00093 & 0.00025  &&  -0.00342 & 0.00042  && -0.01277 & 0.00180  && -0.03663 & 0.01590 \\
\hline
\end{tabular}
\end{table}

\begin{table}[H]
\centering
\caption{Bias and MSE of WMHRE and WMHPE for U(0,1)} \label{t6}
\begin{tabular}{c c c c c c c c}
\hline
 &  && \multicolumn{2}{c}{$M^w_{\alpha}(g;t)$} && \multicolumn{2}{c}{$\bar{M}^w_{\alpha}(g;t)$} \\
 \cline{4-5}  \cline{7-8}
\textbf{n} & \bf{$\alpha$} & \textbf{t} & \textbf{Bias} & \textbf{MSE} &&  \textbf{Bias} & \textbf{MSE} \\
\hline
50  & 0.10 & 0.10 & 0.0054 & 0.0046 && -0.0752 & 0.0065\\
    &     & 0.25 & 0.0179 & 0.0032 && -0.0946 & 0.0115\\
    &     & 0.50 & 0.0142 & 0.0016 && -0.0904 & 0.0141\\
    &     & 0.75 & -0.0271 & 0.0042 && -0.0699 & 0.0109\\
    &     & 0.90 & -0.0843 & 0.0144 && -0.0504 & 0.0084\\[1ex]

    & 0.25 & 0.10 & 0.0156 & 0.0052 && -0.0568 & 0.0036\\
     &     & 0.25 & 0.0242 & 0.0045 && -0.0946 & 0.0115\\  
    &     & 0.50 & 0.0244 & 0.0019 && -0.0891 & 0.0137\\
    &     & 0.75 & -0.0108 & 0.0020 && -0.0763 & 0.0122\\ 
     &     & 0.90 & -0.0468 & 0.0050 && -0.0549 & 0.0092\\[1ex]   

    & 0.50 & 0.10 & 0.0234 & 0.0095 && -0.0403 & 0.0018\\
    &     & 0.25 & 0.0464 & 0.0077 && -0.0838 & 0.0090\\
    &     & 0.50 & 0.0427 & 0.0041 && -0.0940 & 0.0155\\
    &     & 0.75 & 0.0213 & 0.0007 && -0.0829 & 0.0163\\
    &     & 0.90 & -0.0095 & 0.0005 && -0.0615 & 0.0140\\[1ex]

    & 0.75 & 0.10 & 0.0654 & 0.0281 && -0.0365 & 0.0015\\
    &     & 0.25 & 0.0872 & 0.0243 && -0.0921 & 0.0107\\
    &     & 0.50 & 0.0884 & 0.0140 && -0.1181 & 0.0254\\
    &     & 0.75 & 0.0568 & 0.0036 && -0.1195 & 0.0329\\
    &     & 0.90 & 0.0131 & 0.0002 && -0.0919 & 0.0330\\
 \hline   
\end{tabular}
\end{table}

\begin{table}[H]
\centering
\caption{Bias and MSE of WMHRE and WMHPE for Exp(1)} \label{t7}
\begin{tabular}{c c c cc c cc}
\hline
 &  && \multicolumn{2}{c}{$M^w_{\alpha}(g;t)$} && \multicolumn{2}{c}{$\bar{M}^w_{\alpha}(g;t)$} \\
 \cline{4-5}  \cline{7-8}
\textbf{n} & \bf{$\alpha$} & \textbf{t} & \textbf{Bias} & \textbf{MSE} &&  \textbf{Bias} & \textbf{MSE} \\
\hline
50  & 0.10 & 0.10 & -0.0534 & 0.0036 && -0.0409 & 0.0017\\
    &     & 0.25 & -0.0370 & 0.0023 && -0.0933 & 0.0091\\
    &     & 0.50 & -0.0360 & 0.0038 && -0.1339 & 0.0189\\
    &     & 0.75 & -0.0507 & 0.0091 && -0.1392 & 0.0211\\
    &     & 0.90 & -0.0710 & 0.0160 && -0.1334 & 0.0196\\
    &     & 1 & -0.0832 & 0.0235 && -0.1271 & 0.0181 \\
    &     & 1.5 & -0.1609 & 0.0843 && -0.1097 & 0.0138 \\
    &     & 2 & -0.3688 & 0.3828 &&  -0.0950 & 0.0105 \\[1ex]

    & 0.25 & 0.10 & -0.0760 & 0.0073 && -0.0323 & 0.0011\\
     &     & 0.25 & -0.0511 & 0.0038 && -0.0840 & 0.0074\\  
    &     & 0.50 & -0.0487 & 0.0046 && -0.1334 & 0.0187\\
    &     & 0.75 & -0.0644 & 0.0102 && -0.1502 & 0.0243\\ 
     &     & 0.90 & -0.0794 & 0.0162 && -0.1469 & 0.0238\\
     &     & 1 & -0.0901 & 0.0211 && -0.1459 & 0.0237 \\
    &     & 1.5 & -0.1743 & 0.0740 &&  -0.1304 &  0.0201 \\
    &     & 2 & -0.4093 & 0.4203 &&  -0.1196 & 0.0169 \\[1ex]

    & 0.50 & 0.10 & -0.1472 & 0.0301 && -0.0232 & 0.0006\\
    &     & 0.25 & -0.1017 & 0.0167 && -0.0745 & 0.0058\\
    &     & 0.50 & -0.0777 & 0.0085 && -0.1456 & 0.0223\\
    &     & 0.75 & -0.0936 & 0.0117 && -0.1826 & 0.0359\\
    &     & 0.90 & -0.1152 & 0.0188 && -0.1935 & 0.0411\\
    &     & 1 & -0.1332 & 0.1189 && -0.1953 & 0.0426 \\
    &     & 1.5 & -0.3865 & 0.2899 &&  -0.1978 & 0.0470 \\
    &     & 2 & -0.7152 & 0.9133 && -0.1952 & 0.0471 \\[1ex] 

    & 0.75 & 0.10 & -0.3792 & 0.2500 && -0.0215 & 0.0005\\
    &     & 0.25 & -0.2437 & 0.1632 && -0.0886 & 0.0081\\
    &     & 0.50 & -0.1595 & 0.1058 && -0.2064 & 0.0447\\
    &     & 0.75 & -0.1598 & 0.0973 && -0.2894 & 0.0897\\
    &     & 0.90 & -0.1913 & 0.1079 && -0.3252 & 0.1145\\
    &     & 1 & -0.0832 & 0.0235 && -0.3437 & 0.1315 \\
    &     & 1.5 & 0.1609 & 0.0843 && -0.3830 & 0.1805 \\
    &     & 2 & -0.3688 & 0.3828 && -0.3940 & 0.2146   \\
 \hline   
\end{tabular}
\end{table}

\begin{table}[H]
\centering
\caption{Bias and MSE of WMHRE and WMHPE for AR(1) model with $\phi=0.5$, $\alpha=0.50$ and Exp(1.5) errors} \label{t8}
\begin{tabular}{p{1.5cm} p{2cm} p{1cm} p{1cm} p{2cm} p{1.5cm}}
\hline
 & \multicolumn{2}{c}{\bf{$M^w_{\alpha}(g;t)$}} && \multicolumn{2}{c}{\bf{$\bar{M}^w_{\alpha}(g;t)$}}\\
\cline{2-3} \cline{5-6}
\bf{t} & \bf{Bias} & \bf{MSE} && \bf{Bias} & \bf{MSE}\\
\hline
0.10 & -0.6322 & 0.4191 && -0.0353 & 0.0013\\
0.25 & -0.4471 & 0.2160 && -0.1397 & 0.0206\\
0.50 & -0.2147 & 0.0541 && -0.3698 & 0.1466\\
0.75 & -0.0943 & 0.0134 && -0.5605 & 0.3356\\
0.90 & -0.0687 & 0.0113 && -0.6288 & 0.4209\\
   1 & -0.0593 & 0.0125 && -0.6682 & 0.4755\\
 1.5 & -0.1302 & 0.0823 && -0.7107 & 0.5428\\
   2 & -0.4358 & 0.5914 &&  -0.6810 & 0.4952\\
\hline   
\end{tabular}
\end{table}

From Table \ref{t4} and \ref{t5}, we observe that, when sample size increases the bias and the mse of $M^w_{\alpha,n}(X)$ decrease. This phenomenon is obvious for consistent estimators. Note that, when $\alpha<1$ i.e., $1-\alpha>0$, $\hat{g}_n^{1-\alpha}$ has positive powers of kernel estimator. Small values are therefore suppressed. But when $\alpha>1$ i.e., $1-\alpha<0$, $\hat{g}_n^{1-\alpha}=\frac{1}{\hat{g}_n^{\alpha-1}}$. Thus a small error in $\hat{g}_n$ can produce a very large error in the estimator. The problem gets worse as $\alpha\to 2$. So the lower values of $\alpha$ gives better estimates. For $M^w_{\alpha,n}(g;t)$ and $\bar{M}^w_{\alpha}(g;t)$ estimators, we simulate random sample from U(0,1) and Exp(1) distributions and calculate the bias and the MSE for different values of $\alpha$ and $t$. The results are provided in Table \ref{t6} and \ref{t7}, respectively. From the tables , it is observed that, the estimators for residual and past entropy also works well. For dependent data, we simulate samples of size 50 from AR(1) model with correlation coefficient 0.5 and Exp(1.5) errors. The bias and MSE of $M^w_{\alpha,n}(g;t)$ and $\bar{M}^w_{\alpha}(g;t)$ are presented in Table \ref{t8}. Performance is also good for dependent data.


\section{Conclusion}\label{S8} In this paper, we propose weighted Mathai-Haubold entropy and its residual and past versions, studied various properties, developed numerous bounds, construct aging classes and obtained characterization results for Weibull distribution. Also, we introduced non-parametric estimators for the three proposed measures, studied their properties and evaluated their performance using simulation.

Kernel density based estimators are discussed for independent and dependent data set. Other estimators based on spacings, quantiles can be used. More work is needed in this direction.


\pagebreak

\section*{Conflicts of Interest} The authors declare no conflict of interest.

\section*{Funding} No funding is received for this work.

\end{document}